\documentclass[12pt]{article}

\usepackage{hyperref}
\hypersetup{
    colorlinks=True,       
    linkcolor=blue,          
    citecolor=blue,        
    filecolor=magenta,      
    urlcolor=blue           
}

\usepackage{amsfonts,amsmath,amssymb}
\usepackage[in]{fullpage}
\usepackage{graphicx}
\usepackage{dcolumn}
\usepackage{bm}

\usepackage[utf8]{inputenc}
\usepackage[T1]{fontenc}
\usepackage{mathptmx}
\usepackage{etoolbox}
\usepackage{braket}

\usepackage[sorting=none]{biblatex} 
\usepackage{eucal}
\usepackage{bm}
\numberwithin{equation}{section}

\usepackage{amsthm}

\newcommand{\R}{\mathbb{R}}

\newcommand{\N}{\mathbb{N}}

\newtheorem{theorem}{Theorem}[section]

\newtheorem{definition}[theorem]{Definition}

\newtheorem{problem}[theorem]{Problem}

\title{\bf
\rule{\linewidth}{1pt}
Can the PageRank centrality be manipulated to obtain any desired ranking?  \rule{\linewidth}{1pt}
} 

\makeatletter
\renewcommand\@date{{%
  \vspace{-2\baselineskip}%
  \large\centering
  \begin{tabular}{@{}c@{}}
    Gonzalo Contreras-Aso\footnote{Corresponding author: gonzalo.contreras@urjc.es} \textsuperscript{$,\,\sharp$,$\flat$}
  \end{tabular}%
  \quad and \quad
  \begin{tabular}{@{}c@{}}
    Regino Criado\textsuperscript{$\sharp$,$\flat$,$\natural$}
  \end{tabular}%
  \quad and\quad
  \begin{tabular}{@{}c@{}}
    Miguel Romance\textsuperscript{$\sharp$,$\flat$,$\natural$}
  \end{tabular}

  \bigskip

{\it \normalsize
  \textsuperscript{$\sharp$}Departamento de Matem\'atica Aplicada, Ciencia e Ingenier\'{\i}a de los Materiales y Tecnolog\'{\i}a Electr\'onica, Universidad Rey Juan Carlos, 28933 M\'ostoles (Madrid), Spain\par
  \textsuperscript{$\flat$}Laboratory of Mathematical Computation on Complex Networks and their Applications, Universidad Rey Juan Carlos, 28933 M\'ostoles (Madrid), Spain\par
  \textsuperscript{$\natural$}Data, Complex networks and Cybersecurity Research Institute, Universidad Rey Juan Carlos, 28028 Madrid, Spain
}

  \bigskip

  \today
}}
\makeatother

\begin{document}

\maketitle

\begin{abstract}
The significance of the PageRank algorithm in shaping the modern Internet cannot be overstated, and its Complex Network theory foundations continue to be a subject of research. In this article we carry out a systematic study of the structural and parametric controllability of PageRank's outcomes, translating a spectral Graph Theory problem into a geometric one, where a natural characterization of its rankings emerges. Furthermore, we show that the change of perspective employed can be applied to the biplex PageRank proposal, performing numerical computations on both real and synthetic network datasets to compare centrality measures used.
\end{abstract}


\section{Introduction}\label{sec1:introduction}

Almost 25 years have passed since the PageRank algorithm was devised \cite{Page-pagerank-1998}. It brought about two revolutions: on the industry side, it shaped the Internet landscape making Google the giant it is today. On the academic side, it triggered an enormous cascade of studies, interested in understanding its properties, its limitations and its implications \cite{Bianchini-Inside-2005, Langville-deeper-2003, Boldi-pagerank-2005, Boldi-pagerank-2009, Langville-Google-2006}. Furthermore, it has been shown to be relevant beyond its original goal of webpage ranking: indeed, it has found applications in very diverse fields such as biology, engineering and even literature (see \cite{Gleich-pagerank-2015} for an extensive survey).

The academic research poured in the PageRank algorithm coincided with both the development of the interdisciplinary field of Complex Networks and the advent of accessible computing resources. This allowed for both theoretical and numerical results \cite{Langville-Google-2006,Gleich-pagerank-2015, Boldi-pagerank-2005, Bianchini-Inside-2005} that have many direct applications in the economic and social world, since the PageRank algorithm is in the core of the most popular web engines. One of these direct implications in marketing and economics is the so called {\sl Search Engine Optimization Problem}  or  {\sl Web Positioning Problem} \cite{chaffey2009search,ledford2015search,henzinger2001hyperlink}, that tries to find out the strategies that can be performed in a network in order to the maximize PageRank of a specific node (or set of nodes). This theoretical problem has huge real applications with severe economic impact in the global markets. In our nowadays on-line world, for any company not only it is crucial to be present in the WWW, but also to appear highest in the ranking of any web engine; the web-master of a site is thus interested in increasing the PageRank of website by connecting it properly with other webpages, since the highest the ranking of a website the biggest economic revenue the corresponding company gets \cite{olsen2009future}. This major Search Engine Optimization Problem belongs to a more general class of problems related with centrality measures of networks: the control of a centrality measure. This general problem deals with the ability to modify at our wish the centrality of a specific node (or set of nodes) of a given network by slightly changing the link structure of the network or by modifying the intrinsic parameters of the centrality measure. Note that the  Search Engine Optimization problem is related to the control of centrality measures by changing the link structure, while in this paper we will focus on the control by modifying the intrinsic parameters of the centrality measure. 

While Search Engine Optimization Problem has attracted broad attention by the scientific community, the control of a centrality measure by modifying its intrinsic parameters has been less considered despite the fact that it also has some potential real applications, since, for example, it gives valuable information for a web engine administrator about how to modify the ranking of a webpage (or a set of webpages) simply by tuning the parameters of the centrality measure that is behind his web searcher.  It is well known that most of the web engines work with algorithms that modify their ingredients in order to improve the results \cite{Langville-Google-2006}, so a detailed analysis of the influence and sensibility of each parameter of these centrality measures must be considered. In particular, in the case of PageRank centrality, there are two parameters of this measure to be considered\cite{Langville-Google-2006, Gleich-pagerank-2015}: the {\sl damping factor} $\alpha\in(0,1)$ and the {\sl personalization vector} $\bm{v}\in\R^n$. The damping factor has been extensively studied, discussed and interpreted (see e.g. \cite{Boldi-pagerank-2005}), but the role of  personalization vector has always remained understudied \cite{Garcia-localization-2012}. 

In this article, we attempt to shed some light on the relationship between the centrality vectors resulting from PageRank and the choice of personalization vectors. This is actually intertwined with the subject of centrality control in Complex Networks\cite{Latora-controlling-2012}: probing the space of possible centrality vectors with suitable changes in either the underlying graph or the centrality measure. There are already a number of studies discussing the possibility of increasing a node's own PageRank score\cite{Avrachenkov-effect-2006, deKerchove-maximizing-2008, Olsen-maximizing-2010, Carchiolo-long-2018}, as well as some advances regarding PageRank competitors\cite{Garcia-localization-2012} (nodes whose relative ranking position depend on the value of the algorithm's parameters). While these approaches are interesting on their own, they focus on specific nodes and their scores or rankings. In this work we discuss centrality vectors and their rankings as a whole, without reference to individual improvements or detriments.

This paper is structured as follows: In Section~\ref{sec2:notation} we establish some notation and basic graph-theoretical concepts, as well as introduce the terminology that will be used throughout the paper. Section~\ref{sec3:monoplex} presents the mathematical definition of the PageRank algorithm and then explore some routes towards controlling its resulting centrality, with either structural or parametric changes. Theoretical results connecting PageRank and personalization vectors are proven, and network datasets are then used for numerical comparisons and discussion of the implications. In Section~\ref{sec4:biplex} we apply the same techniques to the case of the biplex PageRank~\cite{Pedroche-biplex-2016}, an alternative centrality measure based on the PageRank algorithm. We conclude with a discussion and comparison between the results obtained with each of the different approaches.

\section{Preliminaries and notation}\label{sec2:notation}

Let $G=(V,E)$ be a graph (irregardless of directionality or weights), with node set $V=\{1,...,n\}$, for some $n\in \N$ and adjacency matrix $A=(a_{ij})$ such that
\begin{equation}
    a_{ij} = \begin{cases}
    w_{ij} & \textrm{if } (i,j) \in E,\\
    0 & \textrm{otherwise,}
    \end{cases}
\end{equation}
where $w_{ij}$ is the weight corresponding to edge $(i,j)$, by default $w_{ij}=1$ if unweighted. 

The in-degree (number of incoming links) and out-degree (number of outgoing links) of node $i\in V$ are defined as
\begin{equation}
    \deg_{in}(i) = \sum_{j=1}^n a_{ji},\quad \deg_{out}(i) = \sum_{j=1}^n a_{ij},
\end{equation}
respectively. For undirected graphs we clearly have $\deg_{in}(i)=\deg_{out}(i)$. Nodes in a graph with no outgoing links, i.e. such that $\deg_{out}(i)\neq0,$ are called \textit{dangling} nodes. As will be pointed out later, only networks without dangling nodes will be considered, since similar results can be obtained for the general settings simply by using some standard techniques \cite{Garcia-localization-2012}.

By using these definitions, we can introduce the first ingredient of PageRank, the row-normalized adjacency matrix $P$, which is defined as:
\begin{equation}
    P = (p_{ij}) = \left(\frac{a_{ij}}{\deg_{out}(i)}\right)\in M_{n\times n} (\R).
\end{equation}

In the theory of Markov processes (i.e. memory-less stochastic processes) this matrix is referred to as the ``transition matrix'' of the random walker, as its component $p_{ij}$ provides the probability of transitioning from state $j$ to state $i$. Due to the intrinsic random walk nature of the PageRank algorithm (as discussed in \cite{Rosvall-maps-2007, Lambiotte-ranking-2012}), we will use that notation from now on.

We will denote vectors as $\bm{v}= (v_1,...,v_n)^T\in \R^n$ and the canonical basis of $\R^n$ as $\{\bm{e}_1,...,\bm{e}_n\}$. The vector with 1 in all components will be $\bm{e}=(1,...,1)^T$. $I_n$ will denote the identity matrix. Finally, we will say that a vector is positive if it is positive components-wise, and we will say that it has unit norm if its 1-norm is equal to 1.

\newpage

\section{Standard PageRank}\label{sec3:monoplex}

The best way to introduce the PageRank algorithm\cite{Page-pagerank-1998} is through the lens of a \textit{random walker with random teleportation. Let us forget about the teleportation step for a while and} consider a random walker on a network $G$: starting at node $i$, at each step it will choose an outlink from those available in its current node, with probability proportional to the weight of each outlink. This is a Markovian process, whose steady state gives a measure of the ``centrality'' of each node. In other words, the more the walker passes through node $i$, the more important or central it is.

The PageRank algorithm corresponds to a \textit{personalized} version of this centrality measure, consisting of a biased random walker: with probability $\alpha$ it will follow the previously described rules of standard random walks and with probability $1-\alpha$ it will ``teleport'' or ``jump'' to a random node in the network, with associated probabilities given by a distribution $\bm{v}$, sometimes called the teleportation vector. The mathematical formulation of this idea is the following:

\begin{definition}[PageRank vector]
Let $G$ be a graph with no dangling nodes, $v$ a positive, unit norm vector and $\alpha\in(0,1)$. Then, the PageRank vector of $G$ with damping factor $\alpha$ and personalization vector $\bm{v}$ is the only positive, unit norm vector (i.e. $\bm{\pi}>0$, $|\bm{\pi}|_1=1$) such that satisfying 
\begin{equation}\label{eq-PageRank}
    \bm{\pi}^T = \bm{\pi}^T(\alpha P + (1-\alpha) \bm{e} \bm{v}^T),
\end{equation}
where $P$ is the transition matrix of the graph.
\end{definition}

Note that existence and uniqueness of $\bm{\pi}$ are guaranteed by the classic Perron Theorem, as $\alpha P +(1-\alpha)e v^T$ is a positive matrix (see for example \cite{Meyer-matrix-2001, Langville-Google-2006}).

In what follows we will restrict ourselves to graphs with no dangling nodes. Were there any, they can be dealt with in the usual way\footnote{When there are dangling nodes involved, one can resort to the standard trickery of substituting $P \rightarrow P + \bm{d}^T\bm{u}$, where $\bm{d}\in \R^n$ is the distribution of dangling nodes and $\bm{u}\in \R^n$ is the distribution of imposed outgoing links from them (see, for instance \cite{Garcia-localization-2012}).}. This does not affect the results discussed here, and we will thus omit it for the sake of clarity.

We are interested in understanding the conditions under which an arbitrary stochastic vector can be set to be the PageRank centrality of a given graph. We can state this more formally:
\begin{problem}[PageRank centrality control] \label{pro-PRcontrol}
Can we modify the graph $G=(V,E)$ or the components of the PageRank measure (damping factor or personalization vector) such that an arbitrary positive, unit norm vector $\bm{\pi}_0$ is the PageRank vector?
\end{problem}

Changing the structure of the graph in some way (adding/removing edges, changing weights) would be considered as a \textit{structural} change, whereas changing the parameters of the PageRank measure, such as the damping factor or its personalization vector, would be a \textit{parametric} change. 

In the context of the Eigenvector centrality it was proven~\cite{Latora-controlling-2012} that by a rather mild structural change as changing edge weights, one is able to fully fix the resulting centrality vector at will, so long as the network is directed and strongly connected. In the present case, where we instead deal with the PageRank centrality, things are not that simple, due to the row-normalization of the adjacency matrix: the construction of $P$ normalizes out any weight placed on out-edges coming from nodes with out-degree equal to 1. The simplest way to see this is considering directed rings, as in Figure~\ref{fig:ring-graph}.

\begin{figure}[h]
    \centering
    \includegraphics[width=0.35\textwidth]{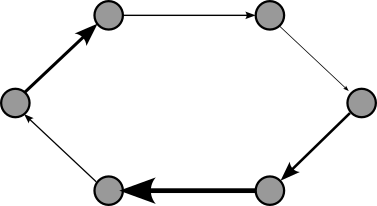}
    \caption{Simple example of a network (the directed cycle $C_6$) whose PageRank centrality is unaffected by any modification of the edge weights.}
    \label{fig:ring-graph}
\end{figure}

We could consider controlling the centrality by means of other types of structural changes, such as adding nodes, rewiring edges, etc. However those are considerably more drastic modifications, and go out of the scope of this paper. Instead, we will now focus on parametric changes, i.e. modifications in the parameters of the centrality measure.

\subsection{Constraints on the personalization vector}

It is clear that suitable adjustments of the damping factor $\alpha$ and the personalization vector $v$ will be needed in order to fix the PageRank centrality of the given network $G$ (see, for example \cite{Boldi-pagerank-2005,Garcia-localization-2012}. What we attempt to do is quantifying the balance between the adjustment of both parameters. In other words, we want to understand what ranges of $\alpha$ provide the desired centrality vector for suitable $v$.   

By operating with \eqref{eq-PageRank} it is straightforward to obtain the following formula \cite{Boldi-pagerank-2005}
\begin{equation}\label{eq-boldi}
    \bm{\pi}^T(I_n - \alpha P) = (1-\alpha) \bm{v}^T.
\end{equation}

Traditionally, this equation can be viewed as an equation for $\bm{\pi}$ given $\alpha$, $\bm{v}$ and $P$. However, we can also view it as an equation for $\bm{v}$ given $\alpha$, $\bm{\pi}$ and $P$:
\begin{equation}\label{eq-personalization}
    \bm{v}^T = \frac{1}{1-\alpha}\bm{\pi}^T(I_n - \alpha P).
\end{equation}

This equation tells us which personalization vector is required to obtain a desired PageRank vector, for a fixed network and damping factor. This raises a question: can we always find such non-negative personalization vector that gets a prescribed PageRank centrality? This natural question is summarized in the following problem:

\begin{problem}[Centrality control via personalization vector] \label{pro-centcontrol}
Given a graph $G$, a damping factor $\alpha\in(0,1)$ and a positive, unit norm vector $\bm{\pi}_0$, does it always exists a positive, unit norm $\bm{v}$ such that the $\bm{\pi}_0$ is the PageRank outcome?
\end{problem}

In other words: can any PageRank vector be set for a given graph and damping factor if we have control over the personalization vector used in the algorithm?

The answer is no, since there is no positive ($v_i>0,\, \forall i$) solution in some cases. Nevertheless, we can study the conditions under which $\bm{\pi}_0$ actually has an associated personalization vector $\bm{v}$ and the following result give a characterization of the existence of positive personalization vectors that give a prescribed PageRank centrality $\bm{\pi}_0$ in terms of the size of its components. 

\begin{theorem}[Existence of the personalization vector]\label{prop-existencev}
Given a graph $G$ and a positive, unit norm vector $\bm{\pi}_0$ then there exists a positive, unit norm personalization vector $\bm{v}$ such that $\bm{\pi}_0$ is the PageRank vector if and only if $\bm{\pi}_0^T \bm{e}_j > \alpha \bm{\pi}_0^T P \bm{e}_j$ for all $j$.
\end{theorem}

\begin{proof}
First we prove that \eqref{eq-personalization} leads to unit norm personalization vectors, since
\begin{align}
    |\bm{v}|_1 &= \bm{v}^T \bm{e} = \frac{1}{1-\alpha}\bm{\pi}_0^T (I_n - \alpha P) \bm{e} = \frac{1}{1-\alpha} \bm{\pi}_0^T (\bm{e}-\alpha P \bm{e}) \nonumber \\
    & = \bm{\pi}_0^T \bm{e} = |\bm{\pi}_0|_1 = 1,
\end{align}
where we used the row-stochasticity in $P\bm{e} = \bm{e}$. We now require that all of $\bm{v}$'s components are positive, so
\begin{equation}
    v_j = \bm{v} \bm{e}_j = \frac{1}{1-\alpha} \bm{\pi}_0^T(I_n -\alpha P)  \bm{e}_j > 0,
\end{equation}
which completes the proof. 
\end{proof}

It is also remarkable to point out that Theorem~\ref{prop-existencev} presents some analytical interplay between the damping factor and personalization vectors, since if we take a positive, unit norm  $\bm{\pi}_0$ and $0<\alpha\leq \min_j (\bm{\pi}^T_0\bm{e}_j)$ then it can be checked that for any graph without dangling nodes there exists a positive, unit norm personalization vector $\bm{v}$ such that $\bm{\pi}_0$ is the PageRank vector. In fact, if we consider a graph without dangling nodes, note that $P \bm{e_j}$ is the $j$-th column of $P$, that is
\begin{equation}
    P \bm{e_j} = \left(\frac{a_{1j}}{\deg_{out}(1)},\frac{a_{2j}}{\deg_{out}(2)},...,\frac{a_{nj}}{\deg_{out}(n)}\right)^T,
\end{equation}
 so  we have that $\bm{\pi}_0^T P \bm{e}_j \leq \bm{\pi}_0^T \bm{e} = 1$, since $0\leq a_{ij}/\deg_{out}(i)\leq 1$ and  hence, if we take $\alpha < \min_j (\bm{\pi}_0^T e_j)$, then
\begin{equation}
    \alpha \bm{\pi}_0^T P e_j \leq \alpha < (\bm{\pi}_0^T e_j) \quad \forall 1\leq j \leq n,
\end{equation}
hence there exists a personalization vector $\bm{v}$ such that $\bm{\pi}_0$ is the PageRank vector, simply by using Theorem~\ref{prop-existencev}.

\subsection{The Ranking control problem}\label{sec:ranking}

In this section we will analyze the centrality control problem by using Theorem~\ref{prop-existencev}, as seen in the previous section.

Centrality measures typically return a list (vector) of centrality scores: numbers between 0 and 1 specifying the importance of each node in the network with respect to the chosen measure. However, for most applications the actual score of a node is not relevant; instead what matters is its relative position with respect to the rest of the nodes. In other words, the ranking of nodes based on their centrality.

The subject of ranking control has remained fairly unexplored due to its technical complexity (as lifting the constraint of fixing concrete centrality vectors makes the problem harder to tackle), but in the PageRank case Theorem~\ref{prop-existencev} provides us with a valuable tool to investigate in this direction by using some techniques from convex geometry.  

Consider the following milder version of Problem~\ref{pro-centcontrol}, where we are now only interested in rankings rather than concrete PageRank vectors.

\begin{problem}[Ranking control via personalization vector]
Given a graph $G$, a damping factor $\alpha\in(0,1)$ and an ordering of the nodes (allowing for ties), does it always exists a positive, unit norm personalization vector $\bm{v}$ such that the PageRank outcome follows the prescribed order? 
\end{problem}

In order to study this problem we will now change the viewpoint of the discussion to a geometric one: consider the $n$-simplex defined as 
\begin{equation}\label{eq-simplex}
    \Delta_n = \{\bm{x} \in \R^n ,\quad \text{such that} \quad \bm{x}>0,\, |\bm{x}|_1=1 \}.
\end{equation}
This set represents the convex span of vectors $\{e_1,...,e_n\}$, and thus is the space of all possible personalization vectors and the space of all possible PageRank vectors of graphs with $n$ nodes. Therefore, we can understand equation \eqref{eq-boldi} as the following map from $\Delta_n$ to itself:
\begin{align}\label{eq-PRmap}
    \bm{\pi}(G, \alpha, \cdot):&\quad  \Delta_n \longrightarrow \Delta_n \nonumber \\
                          &\quad\,\, \bm{v} \lhook\joinrel\longrightarrow \bm{\pi}(G, \alpha, \bm{v})
\end{align}

\begin{figure}[t!]
    \centering
    \includegraphics[width=0.8\linewidth]{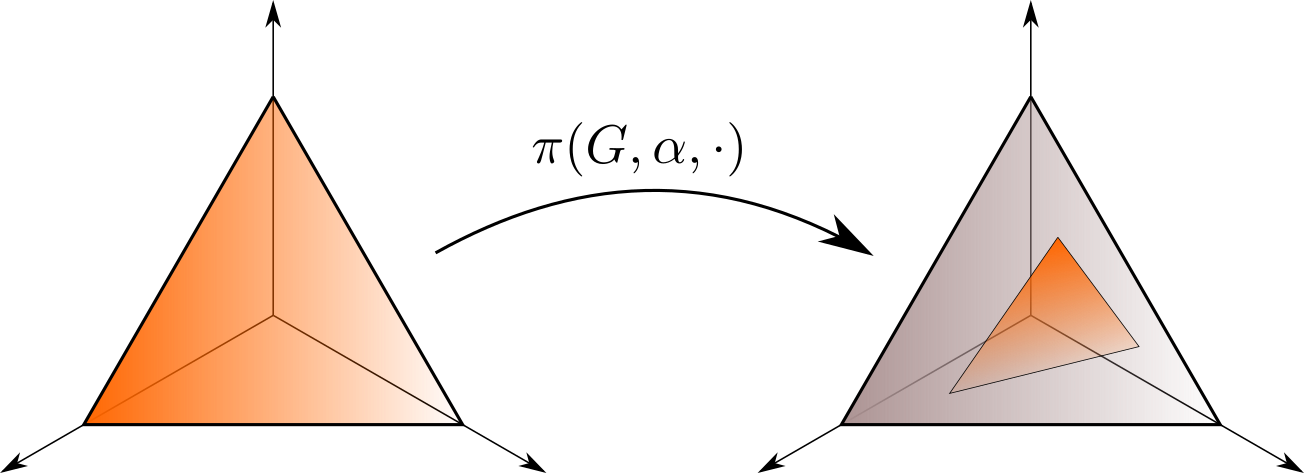}
    \caption{Depiction of the map $\bm{\pi}(G, \alpha, \cdot)$ for $n=3$.}
    \label{fig-map}
\end{figure}

This maps is injective (however in general it is not surjective) and linear in $\bm{v}$, so $\bm{\pi}(G, \alpha, \Delta_n)$ is a polytope (i.e. the convex hull of a finite number of points) in $\Delta_n\subset\R^n$. Figure~\ref{fig-map} illustrates this geometrical interpretation of $\bm{\pi}(G, \alpha, \cdot)$ in case $n=3$.


The key point in this geometric viewpoint is that we can associate each possible ranking to a portion of the simplex. If we consider the center point (barycenter) of the simplex $\Delta_n$, given by the normalization of $\bm{e}$, i.e. $\bm{e}_0\equiv\bm{e}/n=\sum_{i=1}^n \bm{e}_i/n$, then we can define the hyperplanes bisecting the simplex through the center $\bm{e}_0$ and any combination of $n-2$ vertices as
\begin{equation}
    \mathcal{H}_{n}^{i,j} = \left\{\sum_{\substack{k=0 \\ k\neq i,j}}^n \lambda_k \bm{e}_k \quad \text{such that} \quad \lambda_k  \in \R \right\}\subseteq\R^n.
\end{equation}

The relevance of this construction is that it provides us with a way to classify the points $\bm{\pi}\in \Delta_n$ according to their ranking. To see this, consider for instance the hyperplane $\mathcal{H}_4^{1,2}$. It can be identified as the region of ranking space where $c_1=c_2$, by definition. If we move away from it in the direction of $\bm{e}_2$ we will have $c_1<c_2$, and viceversa. 

In general, the ${n \choose 2}$ planes $\mathcal{H}_n^{i,j}$ uniquely determine the pairwise inequalities between components $i,j$ of the PageRank vector. The original simplex $\Delta_n$ is then divided into $n!$ regions (the number of permutations of the components of the PageRank vector), each of them determining a different ranking. A depiction of these regions for the $n=3$ case can be seen in Figure \ref{fig-rankingtriangle}.

\begin{figure}[h!]
    \centering
    \includegraphics[width=0.5\linewidth]{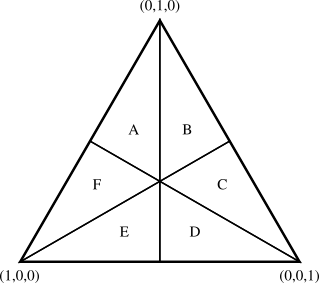}
    \caption{Different ranking regions in the $n=3$ case. For instance, if $\partial A$ denotes the (topological) boundary of $A\subseteq\Delta_n$ in the containing plane of $A$, then   $\bm{\pi} = (\pi_1,\pi_2,\pi_3)\in  A \setminus \partial A$ corresponds to $\pi_2 > \pi_1 > \pi_3$, while the intersection between triangles would lead to equal scores, e.g. $\pi \in B\cap C$ would correspond to  $\pi_2 = \pi_3 > \pi_1$ .}
    \label{fig-rankingtriangle}
\end{figure}

In this light, we can see that there is Ranking control if and only if
\begin{equation}\label{eq-rankingcenter}
\bm{e}_0= \frac{1}{n}\bm{e} \in \textrm{Im}(\bm{\pi}) \quad \text{and}\quad \bm{e}_0 = \frac{1}{n} \bm{e} \notin \partial \textrm{Im}(\bm{\pi}).
\end{equation}

The argument here is identical to that of the hyperplanes: $\bm{\pi}=\bm{e}_0$ is the point in ranking space where $c_1=c_2=...=c_n$. Given that all hyperplanes $\mathcal{H}_n^{i,j}$ pass through $\bm{e}_0$ by construction, all ranking regions are $\varepsilon > 0$ away from it.  Thus, moving $\varepsilon > 0$ away in any direction will lead to different rankings. 

This idea may be easier to visualize if we take into account Figure \ref{fig-map}. Notice that in that case the resulting triangle (right) contains points associated to any ranking (as shown in \ref{fig-rankingtriangle}). The condition necessary and sufficient for this to happen is for it to contain the centerpoint of the bigger triangle.

Next, we can give an analytical characterization of the existence of of a prescribed ranking of the nodes in terms of the relationship between the damping factor and the column sums of $P$, which is the analogous of Theorem~\ref{prop-existencev} but for the Ranking problem.

\begin{theorem}[Characterization of ranking control]\label{thm:ranking}
Given a graph $G$ and damping factor $\alpha=(0,1)$, then it is possible to obtain any ranking of the nodes under the PageRank if and only if
\begin{equation}\label{eq-necessary}
    \frac{1}{\alpha} > \max_j(\sum_{i=1}^n P_{ij}).
\end{equation}
\end{theorem}

\begin{proof}
 Using $\bm{\pi}_0=\bm{e}_0$ in Theorem~\ref{prop-existencev} yields 
\begin{equation}
    \bm{e}_0^T \bm{e}_j = \frac{1}{n} > \alpha \bm{e}_0^T P \bm{e}_j = \frac{1}{n}\alpha   \sum_{i=1}^n P_{ij},
\end{equation}
for all $1\leq j \leq n$,  which already gives us the characterization of the existence of a personalization vector that gives any prescribed ranking of nodes. By virtue of the aforementioned Theorem, we can also conclude it to be a sufficient condition for the existence of a personalization vector allowing for any desired ranking.
\end{proof}

Given that $\sum_i P_{ij}$ is the total probability that a random walker visits node $j$, this theorem can be interpreted as an upper bound for $\alpha$ in terms of the maximum of these total probabilities. This upper bound tells us that, provided we have $\alpha < 1/\max_j\sum_i P_{ij}$, we can always find any desired ranking with an appropriate choice of personalization vector. It is important to note that this is not a statistical result, in the sense that as long as there is one node targeted by many others with low out-degree, there will be almost no room for ranking control, regardless of the topology of the rest of the network. As we will see later, this is very reminiscent of the scale-free~\cite{Barabasi-emergence-1999} network paradigm: indeed, scale-free networks present these high in-degree nodes pointed to by low out-degree ones.

It is also remarkable to point out the fact that if we denote
\begin{equation}\label{eq:ControlPR}
\alpha_0=\frac 1{\max_j\sum_i P_{ij}},
\end{equation}
then $\alpha_0\in(0,1]$ is a measure of the {\sl controllability} of the PageRank in graph $G$, since the bigger $\alpha_0$ is the wider range of damping factors allow Ranking control of PageRank in $G$.


\subsection{Real network datasets}

Having found an network-specific upper bound for the value of the damping factor $\alpha$, which would allow the PageRank of the network to be ranking-controllable tinkering with the personalization vector, it is left for us to find out whether it is a hard or soft bound. 

The standard value considered for the damping is $\alpha = 0.85$ \cite{Langville-Google-2006}, whose interpretation in terms of Internet hyperlink networks is that of an surfer clicking on hyperlinks $\sim$8 times before losing interest and searching for something else; this value corresponds to constraining the maximum of the column sum of $P$ to around 1.17. This is clearly a very strict condition.

In fact, we have computed the maximum of the column sums of $P$ for a variety of networks\footnote{It should be noted that in order to perform this computation we had to deal with the issue of the dangling nodes, which we glossed over at the beginning of this section. We deemed adding a single, random connection from each dangling node to another, non-danging node as the simplest, least intrusive way to remove this issue.}, publicly available from different Internet sources (all fetch from the KONECT network repository \cite{KONECT-datasets} and the CASOS network repository \cite{CASOS-datasets}). We can extract the maximum value of the damping factor $\alpha$ which would enable us to have ranking control over each network's PageRank rankings. This is shown in Figure \ref{fig:real-datasets-alpha}. 

\begin{figure*}[ht!]
    \centering
    \includegraphics[width=1\textwidth]{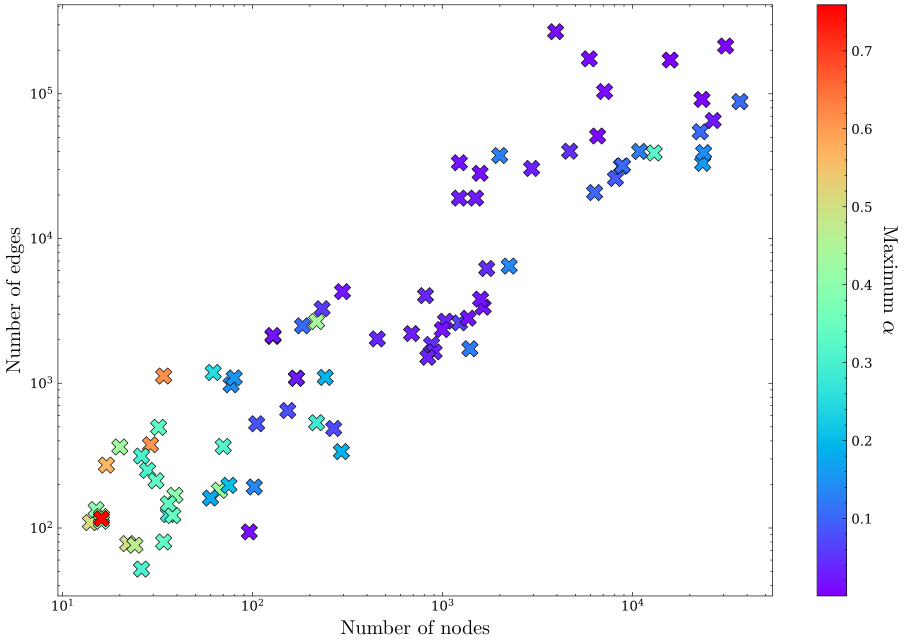}
    \caption{Scatter plot showing the number of edges against the number of nodes for 84 different real networks obtained from the KONECT network repository \cite{KONECT-datasets} and the CASOS network repository\cite{CASOS-datasets}, with datapoints colored based on the maximum value of $\alpha$ providing ranking control.}
    \label{fig:real-datasets-alpha}
\end{figure*}

As expected from the above discussion, the maximum values of the damping factor are generally small compared to the standard $\alpha=0.85$, regardless of the network size. There are a couple of exceptionally high values, but still lower than such value. We see, on the other hand, that the smaller the network the more controllable it is. This can also be understood from Theorem~\ref{thm:ranking}: a higher number of nodes means that the maximum column sum of $P$ is likely to be higher (specially due to the number of edges growing also linearly with the number of nodes), hindering controllability.

While our results are of a theoretical nature, and thus are not related to any specific implementation or application of PageRank, it might also be interesting to address the implications of this bound in some specific use cases of PageRank (and more concretely, understanding the teleportation vector in them).
\begin{itemize}
    \item World Wide Web and similar data: Here the purpose of PageRank is mainly identifying websites of interest for a given user. The teleportation therefore allows for tweaking the preferences of the user, providing different rankings to a user with a different personalization vector. The bound \ref{eq-necessary} in this case tells us that the ordering is robust: the ranking can't be completely altered by the choice of personalization.  

    \item Genetical or Protein-Protein Interaction networks: As explained in \cite{Gleich-pagerank-2015}, PageRank has been applied in a variety of biological networks \cite{2005-Morrison-GeneRank, 2007-Freschi-protein}. In these applications, the teleportation vector is designed to focus the search on specific areas of the network. Given that, as discussed in the aforementioned paper, the damping used in these applications is high ($\alpha \approx 0.8$), we can also to conclude that the ranking will also be robust with respect to changes in the personalization vector.

    \item Knowledge information systems: It is also discussed in \cite{Gleich-pagerank-2015} that PageRank has found a way to be used in semi-supervised learning tasks (for instance, in graphs where each node is an image and two are connected if they share a caption label \cite{2004-Pan-automatic}). These studies present a different phenomenology to the previously discussed cases, as they employ very low values of $\alpha\approx 0.1$. Although we have no datasets of this type, our results force us to conclude that it is quite likely that the rankings obtained will be highly dependent on the personalization vector used.  

    \item There are plenty other applications (see \cite{Gleich-pagerank-2015}), most of them using high $\alpha$. We can draw similar conclusions to the previous cases, the robustness of the ranking.
    
\end{itemize}
Some of the datasets used in Figure \ref{fig:real-datasets-alpha} fall into these or other categories. The specific data used in each of them can be found in our GitHub repository (which can be found in the Data Availability section of this manuscript), and the information about each dataset is in \cite{KONECT-datasets, CASOS-datasets}.

\section{Biplex PageRank}\label{sec4:biplex}

In \cite{Pedroche-biplex-2016} a novel version of the PageRank vector was put forward by establishing an analogy between the standard PageRank algorithm and a random surfer on a ``virtual'' biplex network, constructed from the initial graph $G$.

Although we will not discuss it here, this algorithm was shown to be useful in order to extend the notion of PageRank centralities to multiplex networks. Multiplex networks are networks where the interactions between nodes fall into different categories \cite{Boccaletti-structure-2014}. Hence, they can be represented as different layers, each containing the same nodes but with a particular set of connections. Standard complex network algorithms (such as centrality measures, community detection algorithms, and others) need to be extended to account for these more intricate structures, and the biplex PageRank~\cite{Pedroche-biplex-2016} is one of such proposals. 

Nevertheless, the application of this algorithm to monoplex networks provides yet another extension of PageRank, which can serve as a playground for novel ideas related to centrality. In our case, it will be clear that the geometrical solution to the ranking problem described in Section~\ref{sec:ranking} is not restricted to just the vanilla PageRank algorithm: it can serve as a guiding principle in more complicated, although related, measures.

In the biplex PageRank algorithm, the auxiliar biplex network considered consists of two layers: one with the actual edge connections between the $n$ nodes (this layer essentially accounts for the teleportation-less random walk), while the other contains a fully connected graph between them (this is the ``teleportation layer''). The biased random walker with teleportation then chooses, at each step, whether to follow the links in the usual transition layer, or the teleportation layer.

This construction led to the following definition:
\begin{definition}[Biplex PageRank centrality~\cite{Pedroche-biplex-2016}]
Let $G$ be a graph with no dangling nodes, with transition matrix $P$. Let $\bm{v}$ be a positive, unit norm vector and $\beta\in(0,1)$. Then, the biplex PageRank vector of $G$ with damping factor $\beta$ and personalization vector $\bm{v}$ is the vector
\begin{equation}\label{eq-BiplexPageRank}
    \bm{\pi}_{\rm BPR} = \bm{\pi}_u + \bm{\pi}_d ,
\end{equation}
where $[\bm{\pi}^T_u, \bm{\pi}^T_d]\in\R^{2n}$ is the only positive, unit norm eigenvector of 
\begin{equation}\label{eq_matrixBPR}
    M_{\rm BPR} = \begin{pmatrix}
    \beta P & (1-\beta)I_n \\
    \beta I_n & (1-\beta)\bm{e} \bm{v}^T
    \end{pmatrix}.
\end{equation}
\end{definition}

Note that $\pi_u$ corresponds to the centrality of the nodes in the transition layer, while $\pi_d$ corresponds to the centrality of the nodes in the teleportation layer.

It is remarkable to point out that existence and uniqueness of the Biplex PageRank centrality are granted by the Perron-Frobenius Theorem. This alternative version of the biased walker leads to a different centrality measure, whose technical details we will skip, only keeping the necessary ones and referring the interested reader to~\cite{Pedroche-biplex-2016} for them.

Vectors $\bm{\pi}_u$ and $\bm{\pi}_d$ satisfy the following relations, $\bm{\pi}_u \bm{e} = \beta$ and $\bm{\pi}_{\color{red} d} \bm{e} = 1-\beta$. 

Later in \cite{Pedroche-sharp-2018} a closed form formula for the Biplex PageRank vector, in resemblance to formula \eqref{eq-boldi}, was found as
\begin{equation}\label{eq-piBiplex}
    \bm{\pi}_{\rm BPR}^T = (1-\beta)^2 \bm{v}^T \left(\beta I_n + Y\right)\,Z^{-1},
\end{equation}
where $Y=I_n - \beta P$,  $Z = \gamma I_n - \beta P$ and $\gamma=1-\beta(1-\beta)$. It is straightforward to check that $(\beta I_n+Y)$ is invertible in the $\beta\in(0,1)$ range, so we also have the formula
\begin{equation}\label{eq-vBiplex}
    \bm{v}^T = \frac{1}{(1-\beta)^2}\bm{\pi}_{\rm BPR}^T\,Z\left(\beta I_n + Y\right)^{-1}.
\end{equation}
With this, we can state the following theorem that characterizes when a personalization vector exists for a prescribed biplex PageRank centrality.

\begin{theorem}[Existence of the personalization vector - biplex case]\label{prop-existencevBiplex}
Given a graph $G$ and a positive, unit norm $\bm{\pi}_{BPR}$, then there exists a positive, unit norm personalization vector $\bm{v}$ such that $\bm{\pi}_{0}$ is the biplex PageRank vector if and only if $\bm{\pi}_{BPR}^T \bm{e}_j > \beta \bm{\pi}_{BPR}^T \mathcal{P} \bm{e}_j$ for all $j$, where $\mathcal{P}=(2-\beta)(\beta I_n + Y)^{-1}$.
\end{theorem}

\begin{proof}
First we prove that \eqref{eq-vBiplex} leads to unit-norm personalization vectors. Note that $P^n \bm{e}=\bm{e}$ due to row-stochasticity, therefore if we use the resolvent expansion 
\[
(\beta I_n+Y)^{-1}=\frac{1}{\beta+1}\sum_{m=0}^\infty\left(\frac{\beta}{1+\beta} P\right)^m,
\]
we end up with 
\begin{align*}
    |\bm{v}|_1 &=\bm{v}^T \bm{e} = \frac{1}{(1-\beta)^2}\bm{\pi}_{\rm BPR}^T\,Z\left(\beta I_n + Y\right)^{-1} \bm{e} \\
    &= \frac{1}{(1-\beta)^2}\bm{\pi}_{\rm BPR}^T\,Z \frac{1}{\beta+1}\sum_{m=0}^\infty \left(\frac{\beta}{1+\beta} P \right)^m e \\
    & = \frac{1}{(1-\beta)^2}\bm{\pi}_{\rm BPR}^T\,(\gamma I_n - \beta P) \bm{e} \frac{1}{\beta+1}\sum_{m=0}^\infty \left(\frac{\beta}{1+\beta}\right)^m \\
    &= \frac{1}{(1-\beta)^2}\bm{\pi}_{\rm BPR}^T \, (1 - \beta)^{\color{red}2}\bm{e} = \bm{\pi}_{\rm BPR}^T \bm{e} = 1.
\end{align*}

We now require that all components of the required personalization vector are positive, 
\begin{align}
    v_j &= \bm{v}^T \bm{e}_j = \frac{1}{(1-\beta)^2}\bm{\pi}_{\rm BPR}^T Z (\beta I_n+Y)^{-1} \bm{e}_j > 0.
\end{align}

It will now be convenient expanding the $Z(\beta I_n+Y)^{-1}$ expression in with the previously mentioned resolvent series, multiplying and re-summing. Doing so we find
\begin{align*}
    & Z(\beta I_n+Y)^{-1} = \frac{1}{\beta+1}(\gamma I_n - \beta P) \sum_{m=0} \left(\frac{\beta}{\beta+1}\right)^m P^m \\ 
    &= \frac{1}{\beta+1}\sum_{m=0} \left[\gamma\left(\frac{\beta}{\beta+1}\right)^m P^m - \beta\left(\frac{\beta}{\beta+1}\right)^m P^{m+1} \right] \\
    &= \frac{1}{\beta+1}\left[\gamma I_n - \beta(\beta-2)I_n + \beta(\beta-2) \sum_{m=0}^\infty \left(\frac{\beta}{\beta+1}\right)^{m} P^{m} \right] \\
    &= I_n + \beta(\beta-2)(\beta + Y)^{-1}.
\end{align*}

Plugging this in the above equation we find the condition
\begin{equation}
    \left[\bm{\pi}_{\rm BPR} + \beta(\beta-2)\bm{\pi}_{\rm BPR}(\beta I_n + Y)^{-1}\right]\bm{e}_j > 0
\end{equation}
which, with the identification $\mathcal{P}=(2-\beta)(\beta I_n + Y)^{-1}$ concludes the proof. 
\end{proof}

Again, by using the geometric approach proposed in Section~\ref{sec:ranking} we can interpret the biplex PageRank vector as the linear map between simplices \eqref{eq-simplex}
\begin{align}\label{eq-BPRmap}
    \bm{\pi}_{\rm BPR}(G, \beta, \cdot):&\quad  \Delta_n \longrightarrow \Delta_n \nonumber \\
                          &\quad\,\, \bm{v} \lhook\joinrel\longrightarrow \bm{\pi}_{\rm BPR}(G, \beta, \bm{v})
\end{align}

This map is again injective and linear in $\bm{v}$, and consequently allows us to employ the same kind of argument for the existence of ranking controllability
\begin{equation}\label{eq-rankingcenter-biplex}
\bm{e}_0 = \frac{1}{n} \bm{e} \in \text{Im}(\bm{\pi}_{\rm BPR}),\quad \bm{e}_0 = \frac{1}{n} \bm{e} \notin \partial\text{Im}(\bm{\pi}_{\rm BPR}).
\end{equation}

It is straightforward to find an analytic characterization of ranking control in the biplex PageRank case in terms of the relationship between $\beta$ and the column sums of matrix $P$, simply by following the same reasoning used in the standard PageRank setting. In fact, following similar arguments that those used in the proof of Theorem~\ref{thm:ranking} it can be easily proved the following result:

\begin{theorem}[Characterization of biplex ranking control]\label{thm:biplex}
Given a graph $G$ and a damping factor $\beta=(0,1)$, then it is possible to obtain any ranking of the nodes under the biplex PageRank if and only if
\begin{equation}\label{eq-necessaryBiplex}
    \frac{1}{\beta} > \max_{j}  \left(\sum_{i=1}^N  \mathcal{P}_{ij} \right).
\end{equation}
\end{theorem}

By using the definition of $\mathcal{P}$, the condition that appears in Theorem~\ref{thm:biplex} can be rewritten as follows:
\begin{equation}\label{eq-biplex-minalpha}
    \frac{1}{\beta} >  \max_{j}  \left( \sum_{i=1}^N \mathcal{P}_{ij} \right) = \left(\frac{2-\beta}{1+\beta}\right) \max_{j} \sum_{i=1}^N\left[\left(I_n - \frac{\beta}{1+\beta} P\right)^{-1}\right]_{ij},
\end{equation}
but we cannot expect a more simplified expression of the maximal $\beta$ in terms of $P_{ij}$, since matrix $\mathcal{P}$ depends itself on the damping factor, unlike what happened in the standard PageRank case.

\subsection{Comparison to the monoplex result}

After presenting the analytic result for the Biplex ranking problem (Theorem~\ref{thm:biplex}), we now turn to numerics in order to develop some understanding for this result. In particular, we are interested in how it compares to the usual PageRank, whether it is {\sl more controllable} or not, in terms of the maximal damping factor that allows full ranking control. A similar comparative analysis was performed in terms of the controllability based on the personalization vector between the (classic) PageRank and the Biplex PageRank by J.\,Flores et. al \cite{flores2020parametric}.

As we have pointed out before, Theorem~\ref{thm:ranking} shows that the value $\alpha_0$ introduced in equation\eqref{eq:ControlPR} is a measure of the controllability of the PageRank in $G$. Similarly, if we consider $\beta_0$ the maximal value that verifies equation~\eqref{eq-biplex-minalpha}, then it is also a measure of the controllability of the Biplex PageRank in $G$, since  the bigger $\beta_0$ is the wider range of damping factors allow Biplex Ranking control of PageRank in $G$.  A numerical comparison between $\alpha_0$ and $\beta_0$ for the same real network datasets used in the standard PageRank case (all fetch from the KONECT network repository \cite{KONECT-datasets} and the CASOS network repository \cite{CASOS-datasets}) is presented in Figure~\ref{fig:biplex_vs_PR}. Note that, for most cases the maximal value of the damping factor  $\alpha_0$ is smaller than the corresponding maximal $\beta_0$ for the Biplex PageRank, so we see that in these cases Biplex PageRank is more controllable than (classic) PageRank, which is consistent with the results obtained in~\cite{flores2020parametric} for the controllability related with personalization vectors.

\begin{figure*}[ht!]
    \centering
    \includegraphics[width=1\textwidth]{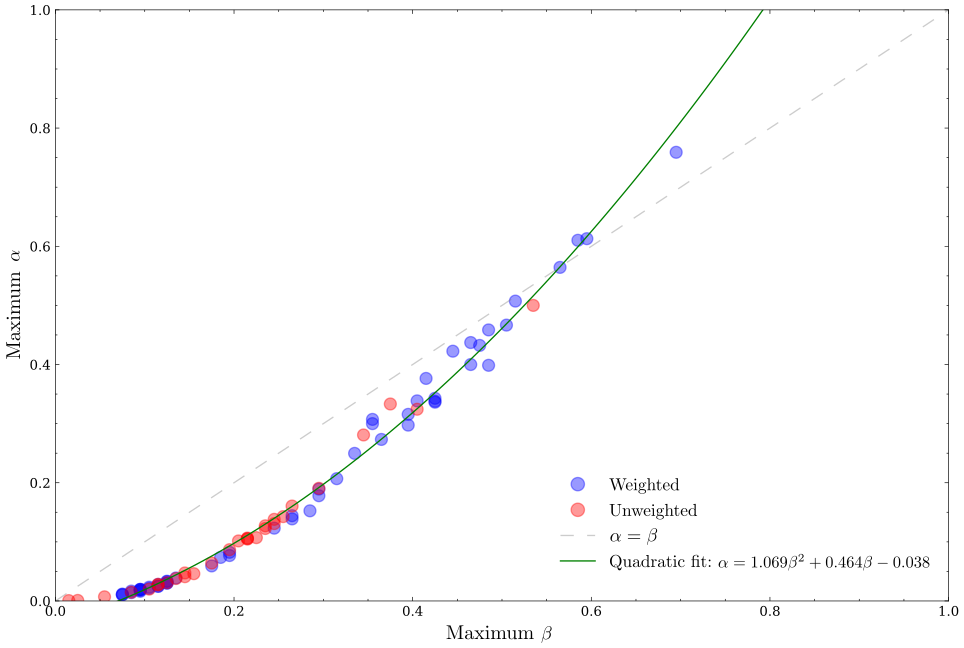}
    \caption{Comparison between maximum $\alpha$ and $\beta$ saturating their respective bounds in the cases of Standard and Biplex PageRank for 84 different real networks obtained from the KONECT network repository \cite{KONECT-datasets} and the CASOS network repository \cite{CASOS-datasets}. Red datapoints represents weighted networks and blue represents unweighted ones.}
    \label{fig:biplex_vs_PR}
\end{figure*}

It is also interesting to point out that, although the datasets come from very heterogeneous sources, there is a clear tendency in the data, following a curve which we found to be (via a quadratic polynomial fit) $y = 1.014 x^2 + 0.492 x - 0.041$. This is perhaps more surprising when we take into account that some of the sampled networks are weighted, yet the quadratic behaviour remains unchanged.

In order to delve deeper in this result, we will consider another batch of network data; this time synthetic networks. We have generated, with the aid of the NetworkX library in Python, two distinct sets of networks: some directed random networks (constructed in the same vein as the undirected Erdös-Renyi version) and some directed scale-free ones (constructed based on the procedure prescribed in \cite{Bollobas-directed-2003}). In both cases we generated networks with the number of nodes ranging from 100 to 20,000, with different edge creation probabilities (see the GitHub repository for the specific details of the implementation). We again compute their maximum values of $\alpha,\beta$ and plot one against the other, obtaining Figure~\ref{fig:BPR_vs_PR-SF_vs_ER}.

\newpage

\begin{figure*}[ht]
    \centering
    \includegraphics[width=1\textwidth]{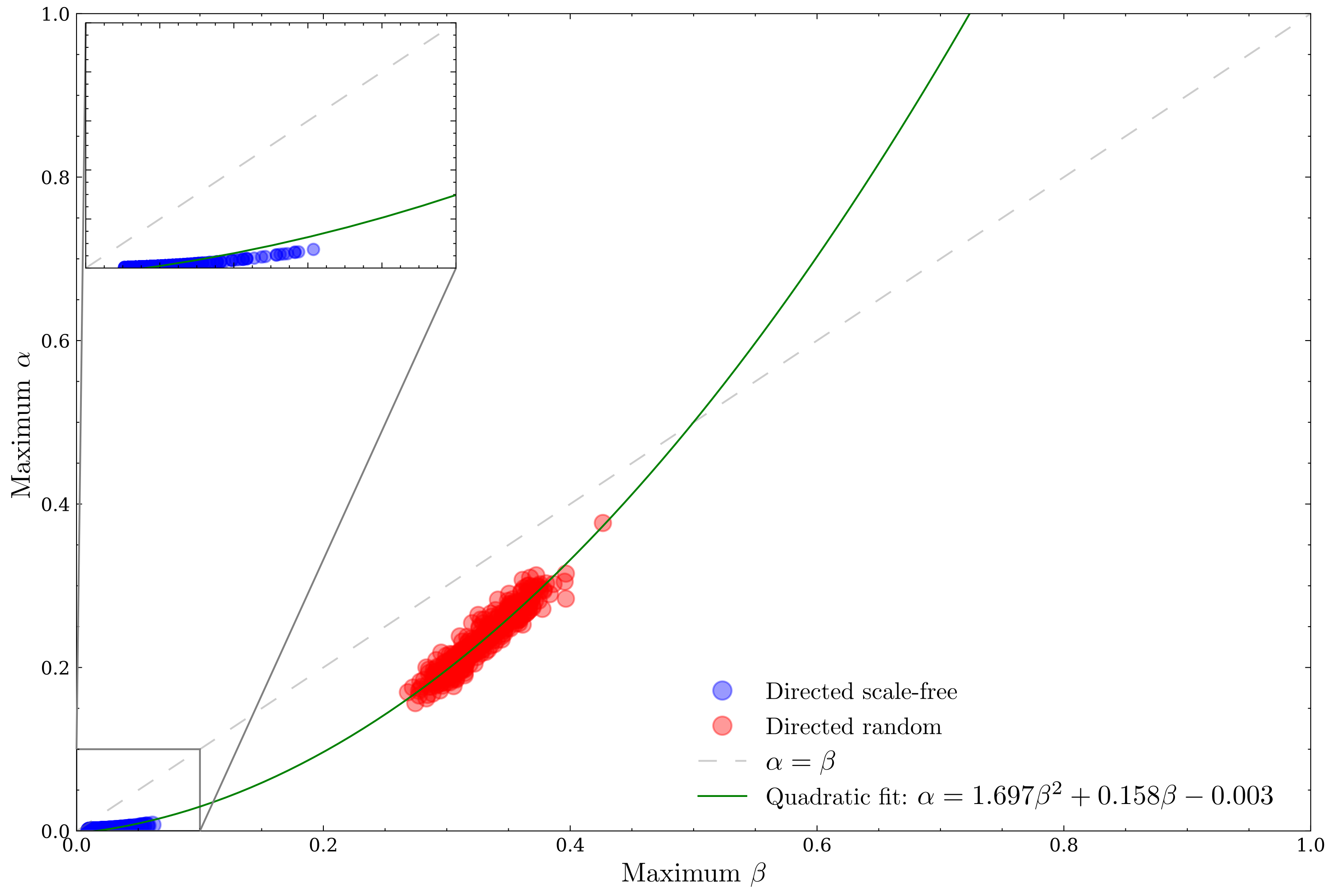}
    \caption{Comparison of maximum $\alpha$ and $\beta$ saturating their respective bounds in the cases of Standard and Biplex PageRank for synthetic networks. Here red datapoints represent random networks and blue represents scale-free ones.}
    \label{fig:BPR_vs_PR-SF_vs_ER}
\end{figure*}

We can clearly see that the polynomial fit is very similar to that of the real networks case, and we could expect them to match even better had we sampled more networks. Apart from that, it is quite noticeable that the distinct nature of the generated networks also separates them in their behavior with respect to both centrality measures. This was already hinted at in the previous section: the presence of high in-degree nodes pointed to by low out-degree ones is very common in scale-free networks, thus their maximum values of $\alpha$ and $\beta$ are specially low. For random networks all nodes have, on average, the same connectivity, thus they allow for more flexibility in ranking control. Another byproduct of the randomness in the corresponding synthetic networks is the higher spread in $\alpha$ for fixed $\beta$ compared to that of scale-free ones.

\section{Conclusions}\label{sec5:conclusions}

Our research has focused on the controllability of the PageRank algorithm as a centrality measure in complex networks. Through our study, we have concluded that full control through weight changes is impossible. Instead, we have investigated the conditions necessary to achieve full control through parametric changes, which involve modifying both the damping factor and the personalization vector. By shifting our focus to centrality rankings rather than centrality scores, we found that a less stringent requirement is sufficient for both standard PageRank and biplex PageRank. However, when we tested this condition on real or synthetic networks, we found it to be a challenging constraint. These findings offer further evidence of the stability of PageRank as an indexing tool.

\section*{Acknowledgements}
This work has been partially supported by projects PGC2018-101625-B-I00 (Spanish Ministry, AEI/FEDER, UE), M1993, M2978 and M3033 (URJC Grants). G.\,C-A. is funded by the URJC fellowship PREDOC-21-026-2164.

\section*{Data Availability Statement}

The data used for the numerical results presented in this work, as well as the code written to analyze and plot them, can be found in the repository  \par \href{https://github.com/LaComarca-Lab/PageRank_CentralityControl}{https://github.com/LaComarca-Lab/PageRank\_CentralityControl}.

\printbibliography

@article{Bianchini-Inside-2005,
    author = {Bianchini, Monica and Gori, Marco and Scarselli, Franco},
    year = {2005},
    month = {02},
    pages = {92-},
    title = {Inside PageRank},
    volume = {5},
    journal = {ACM Trans. Inter. Tech.},
    doi = {10.1145/1052934.1052938}
}

@article{Langville-deeper-2003,
    author = {Langville, Amy and Meyer, Carl},
    year = {2004},
    month = {01},
    pages = {},
    title = {Deeper Inside PageRank},
    volume = {1},
    journal = {Internet Mathematics},
    doi = {10.1080/15427951.2004.10129091}
}

@inproceedings{Boldi-pagerank-2005,
    author = {Boldi, Paolo and Santini, Massimo and Vigna, Sebastiano},
    title = {PageRank as a Function of the Damping Factor},
    year = {2005},
    isbn = {1595930469},
    publisher = {Association for Computing Machinery},
    address = {New York, NY, USA},
    doi = {10.1145/1060745.1060827},
    booktitle = {Proceedings of the 14th International Conference on World Wide Web},
    pages = {557–566},
}

@article{Boldi-pagerank-2009,
    author = {Boldi, Paolo and Santini, Massimo and Vigna, Sebastiano},
    year = {2009},
    month = {01},
    pages = {},
    title = {PageRank: Functional dependencies},
    volume = {27},
    journal = {ACM Trans. Inf. Syst.}
}

@article{Gleich-pagerank-2015,
author = {Gleich, David F.},
title = {PageRank Beyond the Web},
journal = {SIAM Review},
volume = {57},
number = {3},
pages = {321-363},
year = {2015},
doi = {10.1137/140976649}
}

@article{Pedroche-biplex-2016,
author = {Pedroche,Francisco  and Romance,Miguel  and Criado,Regino },
title = {A biplex approach to PageRank centrality: From classic to multiplex networks},
journal = {Chaos: An Interdisciplinary Journal of Nonlinear Science},
volume = {26},
number = {6},
pages = {065301},
year = {2016},
doi = {10.1063/1.4952955},
url = {https://doi.org/10.1063/1.4952955}
}

@article{Pedroche-sharp-2018,
title = {Sharp estimates for the personalized Multiplex PageRank},
journal = {Journal of Computational and Applied Mathematics},
volume = {330},
pages = {1030-1040},
year = {2018},
issn = {0377-0427},
doi = {https://doi.org/10.1016/j.cam.2017.02.013},
url = {https://www.sciencedirect.com/science/article/pii/S0377042717300717},
author = {Francisco Pedroche and Esther García and Miguel Romance and Regino Criado}
}

@article{henzinger2001hyperlink,
  title={Hyperlink analysis for the web},
  author={Henzinger, Monika R},
  journal={IEEE Internet computing},
  volume={5},
  number={1},
  pages={45--50},
  year={2001},
  publisher={IEEE}
}

@article{olsen2009future,
  title={A future in directing online traffic},
  author={Olsen, Patricia R},
  journal={The New York Times},
  pages={10},
  year={2009}
}

@book{ledford2015search,
  title={Search engine optimization bible},
  author={Ledford, Jerri L},
  volume={584},
  year={2015},
  publisher={John Wiley \& Sons}
}

@article{chaffey2009search,
  title={Search engine optimization-best practice guide},
  author={Chaffey, D and Lake, C and Friedlein, A},
  journal={Econsultancy. com Ltd},
  year={2009}
}

@article{flores2020parametric,
  title={Parametric controllability of the personalized PageRank: Classic model vs biplex approach},
  author={Flores, Julio and Garc{\'\i}a, Esther and Pedroche, Francisco and Romance, Miguel},
  journal={Chaos: An Interdisciplinary Journal of Nonlinear Science},
  volume={30},
  number={2},
  pages={023115},
  year={2020},
  publisher={AIP Publishing LLC}
}

@article{Avrachenkov-effect-2006,
author = { Konstantin   Avrachenkov  and  Nelly   Litvak },
title = {The Effect of New Links on Google Pagerank},
journal = {Stochastic Models},
volume = {22},
number = {2},
pages = {319-331},
year  = {2006},
publisher = {Taylor & Francis},
doi = {10.1080/15326340600649052}
}

@article{deKerchove-maximizing-2008,
title = {Maximizing PageRank via outlinks},
journal = {Linear Algebra and its Applications},
volume = {429},
number = {5},
pages = {1254-1276},
year = {2008},
note = {Special Issue devoted to selected papers presented at the 13th Conference of the International Linear Algebra Society},
issn = {0024-3795},
doi = {https://doi.org/10.1016/j.laa.2008.01.023},
url = {https://www.sciencedirect.com/science/article/pii/S0024379508000529},
author = {Cristobald {de Kerchove} and Laure Ninove and Paul {van Dooren}}
}

@InProceedings{Olsen-maximizing-2010,
author="Olsen, Martin",
editor="Calamoneri, Tiziana
and Diaz, Josep",
title="Maximizing PageRank with New Backlinks",
booktitle="Algorithms and Complexity",
year="2010",
publisher="Springer Berlin Heidelberg",
address="Berlin, Heidelberg",
pages="37--48",
isbn="978-3-642-13073-1"
}

@article{Latora-controlling-2012,
       author = {{Nicosia}, V. and {Criado}, R. and {Romance}, M. and {Russo}, G. and {Latora}, V.},
        title = "{Controlling centrality in complex networks}",
      journal = {Scientific Reports},
         year = 2012,
        month = jan,
       volume = {2},
          eid = {218},
        pages = {218},
          doi = {10.1038/srep00218},
archivePrefix = {arXiv},
       eprint = {1109.4521},
 primaryClass = {physics.soc-ph}
}

@InProceedings{Carchiolo-long-2018,
author="Carchiolo, V.
and Grassia, M.
and Longheu, A.
and Malgeri, M.
and Mangioni, G.",
editor="Del Ser, Javier
and Osaba, Eneko
and Bilbao, Miren Nekane
and Sanchez-Medina, Javier J.
and Vecchio, Massimo
and Yang, Xin-She",
title="Long Distance In-Links for Ranking Enhancement",
booktitle="Intelligent Distributed Computing XII",
year="2018",
publisher="Springer International Publishing",
address="Cham",
pages="3--10",
isbn="978-3-319-99626-4"
}

@article{Garcia-localization-2012,
title = {On the localization of the personalized PageRank of complex networks},
journal = {Linear Algebra and its Applications},
volume = {439},
number = {3},
pages = {640-652},
year = {2013},
note = {Special Issue in Honor of Harm Bart},
issn = {0024-3795},
doi = {https://doi.org/10.1016/j.laa.2012.10.051},
url = {https://www.sciencedirect.com/science/article/pii/S0024379512007835},
author = {E. García and F. Pedroche and M. Romance}
}

@inproceedings{KONECT-datasets,
	title = {{KONECT} -- {The} {Koblenz} {Network} {Collection}},
	author = {J\'{e}r\^{o}me Kunegis},
	year = {2013},
	booktitle = {Proc. Int. Conf. on World Wide Web Companion},
	pages = {1343--1350},
	url = {http://dl.acm.org/citation.cfm?id=2488173},
}

@online{CASOS-datasets,
  title = {CASOS network datasets},
  howpublished = {\url{http://www.casos.cs.cmu.edu/tools/data2.php}}
}

@article{2005-Morrison-GeneRank,
author = {Morrison, Julie and Breitling, Rainer and Higham, Desmond and Gilbert, David},
year = {2005},
month = {02},
pages = {233},
title = {GeneRank: Using search engine technology for the analysis of microarray experiments},
volume = {6},
journal = {BMC bioinformatics},
doi = {10.1186/1471-2105-6-233}
}

@INPROCEEDINGS{2007-Freschi-protein,
  author={Freschi, Valerio},
  booktitle={2007 IEEE 7th International Symposium on BioInformatics and BioEngineering}, 
  title={Protein function prediction from interaction networks using a random walk ranking algorithm}, 
  year={2007},
  volume={},
  number={},
  pages={42-48},
  doi={10.1109/BIBE.2007.4375543}
}

@inproceedings{2004-Pan-automatic,
author = {Pan, Jia-Yu and Yang, Hyung-Jeong and Faloutsos, Christos and Duygulu, Pinar},
title = {Automatic Multimedia Cross-Modal Correlation Discovery},
year = {2004},
isbn = {1581138881},
publisher = {Association for Computing Machinery},
address = {New York, NY, USA},
url = {https://doi.org/10.1145/1014052.1014135},
doi = {10.1145/1014052.1014135},
booktitle = {Proceedings of the Tenth ACM SIGKDD International Conference on Knowledge Discovery and Data Mining},
pages = {653–658},
numpages = {6},
location = {Seattle, WA, USA},
series = {KDD '04}
}

@inproceedings{Page-pagerank-1998,
  author = {Page, L. and Brin, S. and Motwani, R. and Winograd, T.},
  booktitle = {Proceedings of the 7th International World Wide Web Conference},
  pages = {161--172},
  organization = "",
  title = {The PageRank citation ranking: Bringing order to the Web},
  year = 1998
}

@ARTICLE{Barabasi-emergence-1999,
       author = {{Barab{\'a}si}, Albert-L{\'a}szl{\'o} and {Albert}, R{\'e}ka},
        title = "{Emergence of Scaling in Random Networks}",
      journal = {Science},
         year = 1999,
        month = oct,
       volume = {286},
       number = {5439},
        pages = {509-512},
          doi = {10.1126/science.286.5439.509},
archivePrefix = {arXiv},
       eprint = {cond-mat/9910332},
 primaryClass = {cond-mat.dis-nn}
}

@article{Boccaletti-structure-2014,
    title = {The structure and dynamics of multilayer networks},
    journal = {Physics Reports},
    volume = {544},
    number = {1},
    pages = {1-122},
    year = {2014},
    issn = {0370-1573},
    doi = {https://doi.org/10.1016/j.physrep.2014.07.001},
    author = {S. Boccaletti and G. Bianconi and R. Criado and C.I. {del Genio} and J. Gómez-Gardeñes and M. Romance and I. Sendiña-Nadal and Z. Wang and M. Zanin}
}

@inproceedings{Bollobas-directed-2003,
  title={Directed scale-free graphs},
  author={B{\'e}la Bollob{\'a}s and Christian Borgs and Jennifer T. Chayes and Oliver Riordan},
  booktitle={ACM-SIAM Symposium on Discrete Algorithms},
  year={2003}
}

@article{Rosvall-maps-2007,
author = {Martin Rosvall  and Carl T. Bergstrom },
title = {Maps of random walks on complex networks reveal community structure},
journal = {Proceedings of the National Academy of Sciences},
volume = {105},
number = {4},
pages = {1118-1123},
year = {2008},
doi = {10.1073/pnas.0706851105},
URL = {https://www.pnas.org/doi/abs/10.1073/pnas.0706851105},
eprint = {https://www.pnas.org/doi/pdf/10.1073/pnas.0706851105}
}

@article{Lambiotte-ranking-2012,
  title = {Ranking and clustering of nodes in networks with smart teleportation},
  author = {Lambiotte, R. and Rosvall, M.},
  journal = {Phys. Rev. E},
  volume = {85},
  issue = {5},
  pages = {056107},
  numpages = {9},
  year = {2012},
  month = {May},
  publisher = {American Physical Society},
  doi = {10.1103/PhysRevE.85.056107},
  url = {https://link.aps.org/doi/10.1103/PhysRevE.85.056107}
}

@book{Langville-Google-2006,
  author = {{Langville, Amy N. ; Meyer}, Carl D.},
  booktitle = { Google's PageRank and beyond },
  isbn = {0-691-12202-4},
  note = { Amy N. Langville and Carl D. Meyer },
  publisher = { Princeton Univ. Press },
  publisherplace = {Princeton [u.a.]},
  title = { Google's PageRank and beyond, the science of search engine rankings },
  year = { 2006 }
}

@book{Meyer-matrix-2001,
  author = {Meyer, Carl D.},
  publisher = {SIAM},
  title = {Matrix Analysis and Applied Linear Algebra},
  url = {http://www.matrixanalysis.com/},
  year = 2001
}

\end{document}